\documentclass{llncs}
\usepackage[cp1251]{inputenc}
\usepackage{amsmath}
\usepackage{amssymb}
\usepackage{amsfonts}

\newcommand{\rank}{{\tt rank}}
\newcommand{\len}{{\tt len}}
\newcommand{\HW}[1]{HW({\rm #1})}
\newcommand{\dleq}{\sqsubseteq}
\newcommand{\bt}{{\mathcal T}}
\newcommand{\D}{{{\mathcal D}}}
\newcommand{\tuple}[1]{\langle#1\rangle}

\newcommand{\lldots}{,\ldots ,}

\newcommand{\QED}{$\square$ \medskip}

\newcommand{\explst}{{\epsilon}}

\newcommand{\EXP}[2]{{{\tt{#1}}{\tt{exp}}({#2})}}

\newcommand{\complclass}[1]{\textsf{#1}}

\newcommand{\PTIME}{\complclass{P}}

\newcommand{\NP}{\complclass{NP}}

\newcommand{\PSPACE}{\complclass{PSpace}}

\newcommand{\EXPTIME}{\complclass{ExpTime}}
\newcommand{\NEXPTIME}{\complclass{NExpTime}}
\newcommand{\NmanyEXPTIME}[1]{\complclass{N#1}\complclass{ExpTime}}
\newcommand{\DEXPTIME}{\complclass{2ExpTime}}
\newcommand{\NDEXPTIME}{\complclass{N2ExpTime}}
\newcommand{\manyEXPTIME}[1]{\textsf{#1\EXPTIME}}

\begin{document}

\title{The Expressiveness of Looping Terms in the Semantic Programming \thanks{The authors were supported by the Russian Science Foundation (Grant No. 17-11-01176)}
}

\author{Sergey Goncharov\inst{1} \and Sergey Ospichev\inst{1} \and\newline Denis Ponomaryov\inst{1,2} \and  Dmitri Sviridenko\inst{1}}

\institute{Sobolev Institute of Mathematics \and Ershov Institute of Informatics Systems, \newline Novosibirsk State University}

\maketitle

\begin{abstract}
We consider the language of $\Delta_0$-formulas with list terms interpreted over hereditarily finite list superstructures. We study the complexity of reasoning in extensions of the language of $\Delta_0$-formulas with non-standard list terms, which represent bounded list search, bounded iteration, and bounded recursion. We prove a number of results on the complexity of model checking and satisfiability for these formulas. In particular, we show that the set of $\Delta_0$-formulas with bounded recursive terms true in a given list superstructure $HW(\mathcal{M})$ is non-elementary (it contains the class $\manyEXPTIME{k}$, for all $k\geqslant 1$). For $\Delta_0$-formulas with restrictions on the usage of iterative and recursive terms, we show lower complexity. 
\end{abstract}


\section{Introduction}
In \cite{Sigma1,Sigma2} a paradigm of the Semantic Programming has been proposed, which paved the way for a new generation of declarative programming languages. The approach of the Semantic Programming allows to abstract away from details of implementation and to focus on the desired properties of a software system under development. It also includes imperative features, which allow to specify the order of computations, when it is necessary. Semantic Programming rests on the computability theory put in terms of $\Sigma$-definability in hereditarily finite list superstructures (see, e.g., \cite{AleksandrovaBazhenov} for recent results). The concept of \textit{list} appears to be general enough to represent common datatypes of programming languages. A program in the Semantic Programming is a formula over a signature, which includes basic list functions such as concatenation, adding an element to a list, taking head or tail of a list, as well as predicates, which can be used to refer to elements and initial segments of lists. Computation is implemented in two ways. The first one is based on testing whether a formula is true in an appropriate list superstructure and is conceptually close to the idea of Model Checking in the field of Software Verification. The second way is via deciding the entailment of a formula from an appropriate theory, which axiomatizes properties of list superstructures. The latter approach is close to the idea of Logic Programming and Deductive Verification. The language of the Semantic Programming is powerful enough to formulate statements about syntactic and semantic properties of programs, thus providing a unified framework for program specification, validation, and verification. The approach has numerous applications in model-driven software engineering and in particular, for the development of AI tools. One of the recent use cases is the application of the Semantic Programming to document processing in the area of Business Process Management \cite{MantsivodaPonomaryov}.

As a trade-off between the expressiveness and computational efficiency a number of sublanguages of the Semantic Programming have been proposed. Of the most important ones is the language of $\Delta_0$-formulas, in which only bounded quantification over elements and initial segments of lists can be used. In \cite{CondTerms}, this language has been extended with conditional list terms, which implement the ``if-then-else'' primitive of programming languages. It has been noted in \cite{OspichevPonomaryov} that there are at least two sources of the computational complexity of model checking for $\Delta_0$-formulas and their extensions. The first one is the complexity of computing list terms in a given list superstructure and the second one is the form of the quantifier prefix. It has been shown that for any polynomially computable structure, there exists a polynomially computable representation of its hereditarily finite list superstructure with the above mentioned basic list functions. Thus, the basic list terms make no contribution to the complexity of model checking (provided it is polynomial or superpolynomial). The same has been shown for conditional terms. The authors have noted the natural relationship with the truth problem for Quantified Boolean Formulas, which provided complexity bounds on the model checking for $\Delta_0$-formulas with conditional terms and general or restricted quantifier prefix.   

In \cite{Doklady,Recursive}, the language of $\Delta_0$-formulas has been further extended to address primitives, which implement looping in programming languages. The authors have introduced non-standard list terms, which correspond to bounded list search, bounded list recursion, and bounded iteration. In this paper, we study the complexity of the language of $\Delta_0$-formulas extended with these non-standard terms, with the goal to describe their contribution to the complexity of the basic reasoning tasks. Naturally, the first task we consider is model checking, i.e., given a formula $\varphi$ and a list superstructure $HW(\mathcal{M})$, decide whether $HW(\mathcal{M})\models\varphi$. The second one is satisfiability, i.e., for a formula $\varphi$ decide whether it is true in some list superstructure. The complexity of this  task is obviously important for program validation, since it indicates how complex it is to identify modelling errors, which might cause inconsistency of the constructed program. To make the contribution of the non-standard terms vivid, we omit two features of the language of $\Delta_0$-formulas, which on their own may cause an increased complexity of the reasoning tasks. First of all, we assume that the language contains constants, basic list functions, but \emph{no other functions}. Second, the range of any variable under a (bounded) quantifier must be given by a \emph{ground} list term. These restrictions are implicitly present in \cite{OspichevPonomaryov}. Our results evidence that the language with the non-standard list terms, which implement bounded recursion or iteration, is more expressive than the language of $\Delta_0$-formulas under the mentioned restrictions. In particular, the complexity of reasoning is increased, which is due to the ability of non-standard terms to succinctly represent lists of large size. 

\section{Preliminaries}
We assume the reader is familiar with basics of the complexity theory. We introduce below the complexity classes mentioned in this paper; all the necessary details can be found in \cite{AroraBarak,Papadimitriou}. 

For a finite alphabet $\Sigma$, let $\Sigma^\ast$ be the set of
all words over $\Sigma$ and for a subset $A\subseteq\Sigma^\ast$, let $f: A\rightarrow\Sigma^\ast$ be a function.
$f$ is said to be \PTIME-computable/\NP-computable if there is a
deterministic/non-deterministic Turing Machine (TM) $T$, respectively, and a polynomial
$p$ such that for any $x\in A$ the value of $f(x)$ can be computed by
$T$ in at most $p(|x|)$ steps, where $|x|$ is the length of the word
$x$. The function $f$ is \PSPACE-computable if for any $x\in
A$ the value of $f(x)$ can be computed by a deterministic TM $T$ using at most $p(|x|)$
cells of the tape of $T$, where $p$ is a polynomial.

For a given $n\geqslant 0$, let $\EXP{1}{n}$ be the notation for $2^n$ and for $k\geqslant 1$, let $\EXP{(k+1)}{n}=2^\EXP{k}{n}$.

For $k\geqslant 1$, the function $f$ is called $\manyEXPTIME{k}$/$\NmanyEXPTIME{k}$-computable, respectively, if for any word $x\in A$, the value of $f(x)$ can be computed by a deterministic/ non-deterministic TM, respectively, in at most $\EXP{k}{p(|x|)}$ steps, where $p$ is a polynomial. 

\smallskip

We slightly abuse the terminology and define $C\in \{\PTIME/\NP/\PSPACE/\manyEXPTIME{k}/$ $\NmanyEXPTIME{k}\}$ as the class of subsets $A\subseteq\Sigma^\ast$ such that their characteristic function is $C$-computable. A subset $A\subseteq\Sigma^\ast$ is called \textit{$C$-hard} if any set from $C$  is $m$-reducible to $A$ by a \PTIME-computable function. A subset $A$ is called \textit{$C$-complete} if it is $C$-hard and is contained in $C$.

A structure $\mathcal{M}$ is \PTIME-computable if so are the functions of $\mathcal{M}$, the characteristic functions of predicates, and the domain of $\mathcal{M}$.

\medskip

In this paper, we define the language of \textit{$\Delta_0$-formulas} as a first-order language with sorts ``urelement'' and ``list'' , in which only bounded quantification of the following form is allowed:
\begin{itemize}
\item a restriction onto the list elements $\forall x\in t$ and $\exists x\in t$;
\item a restriction onto the initial segments of lists $\forall x\sqsubseteq t$ and $\exists x\sqsubseteq t$.
\end{itemize}
where $t$ is a variable-free list term (i.e., it does not contain variables). A list term is defined inductively via constant lists, variables of sort ``list'', and list functions given below.
A constant \textit{list} (which can be nested) is built over constants of sort ``urelement'' (called ``urelements'', for short) and a constant $nil$ of sort ``list'', which represents the empty list. The \textit{list functions} are:
\begin{enumerate}
%\item $nil$ -- the constant which represents the empty list; \smallskip
\item $head$ -- the last element of a non-empty list and $nil$, otherwise; %\smallskip
\item $tail$ -- the list without the last element, for a non-empty list, and $nil$, otherwise; %\smallskip
\item $cons$ -- the list obtained by adding a new last element to a list; %\smallskip
\item $conc$ -- concatenation of two lists; %\smallskip
%\item $\in$ -- the predicate ``to be an element of a list''; \smallskip
%\item $\sqsubseteq$ -- the predicate``to be an initial segment of a list".
\end{enumerate}

\smallskip

We assume that the language contains constants, list functions, but no other functions. The predicates $\in, \sqsubseteq$ are allowed to appear in $\Delta_0$-formulas without any restrictions, i.e., they can be used in bounded quantifiers and atomic formulas. 

\medskip

\noindent $\Delta_0$-formulas are interpreted over hereditarily finite \textit{list superstructures} $HW(\mathcal{M})$, where $\mathcal{M}$ is a structure. Urelements are interpreted as distinct elements of the domain of $\mathcal{M}$ and lists are interpreted as lists over urelements and the distinguished ``empty list'' $nil$. In particular, the following equations hold in every $HW(\mathcal{M})$ (the free variables below are assumed to be universally quantified):

\begin{equation*}
\begin{aligned}
& \neg \exists x \ x\in nil \\
& cons(x,y) = cons(x',y') \rightarrow x=x' \wedge y=y' \\
& tail(cons(x,y)) = x, \ \ head(cons(x,y)) = y \\
& tail(nil) = nil, \ \ head(nil) = nil \\
& conc(nil, x) = conc(x, nil) = x \\
& cons(conc(x,y), z) = conc(x, cons(y,z)) \\
& conc(conc(x,y), z) = conc(x, conc(y,z))
\end{aligned}
\end{equation*}

\medskip

It was shown in \cite{OspichevPonomaryov} that for any \PTIME-computable structure $\mathcal{M}$, there exists a \PTIME-computable representation of its superstructure of finite
lists $HW(\mathcal{M})$, in which the value of any list term is given by a  \PTIME-computable function. In this paper, we omit subtleties related to the representation of models and we simply assume that any $HW(\mathcal{M})$ mentioned in the paper is \PTIME-computable and so is the value of any (standard) list term in $HW(\mathcal{M})$. Since we assumed that no functions except constants and list functions are in the language of $\Delta_0$-formulas, it follows that the set of variable-free $\Delta_0$-formulas true in a given structure $HW(\mathcal{M})$ is \PTIME-computable. In turn, the set of variable-free $\Delta_0$-formulas satisfiable in some structure $HW(\mathcal{M})$ is $\NP$-complete, which is due to the correspondence with the satisfiability problem for propositional boolean formulas. In particular, the upper bound is shown  as follows. Given a variable-free $\Delta_0$-formula $\varphi$, every list term in $\varphi$ is replaced with its value, a constant list; by Lemma 2 in \cite{OspichevPonomaryov}, this transformation can be done by a  \PTIME-computable function. Next, every ground atom $s \propto t$, where $s,t$ are constant lists and $\propto\in\{\in,\sqsubseteq\}$, is evaluated as $true$/$false$ and every equality $s=t$ is replaced with $true$ or $false$ if $s$, $t$ are equal or non-equal, respectively. The resulting formula can again be obtained as a value of a \PTIME-computable function. Finally, every ground atom $P(s_1\lldots s_k)$ is replaced with a boolean variable $x_{P(s_1\lldots s_k)}$, which gives a propositional boolean formula, which is satisfiable iff so is the initial $\Delta_0$-formula $\varphi$.

\section{Looping Terms}
We consider extensions of the language of $\Delta_0$-formulas with \textit{bounded search terms} (or \textit{b-search terms}, for short), \textit{recursive terms}, and \textit{iterative terms} of sort ``list''. The corresponding language extensions are denoted as $\Delta_0$+bSearch, $\Delta_0$+Rec, $\Delta_0$+ Iteration. By default, any formula or a list term in the language of $\Delta_0$-formulas is a formula/list term in these language extensions. \textit{Non-standard list terms} are defined as follows.

If $t(\overline{v})$ and $\theta(\overline{v}, x)$ is a $\Delta_0$+bSearch list term and formula, respectively, then the expression $bSearch_{\in}(\theta, t)(\overline{v})$ and $bSearch_{\sqsubseteq}(\theta, t)(\overline{v})$ is a \textit{b-search term}. It is equal to the first element/initial segment $a$ of $t(\overline{v})$, respectively, such that $\theta(\overline{v}, a)$ holds and it is equal to $t(\overline{v})$, otherwise (i.e., if there is no such $a$). 

%List terms in $\Delta_0$+Iteration are defined as follows. Any $\Delta_0$ term is a $\Delta_0$+Iteration list term. 
If $f(\overline{v}), h(\overline{v},y)$ are $\Delta_0$+Iteration list terms and $i$ is a natural number given in either unary or binary representation, then the expression $<i>Iter[f,h](\overline{v})$ is an \textit{iterative term} and its value is given by $g^i(\overline{v})$ with the following definition:
\begin{itemize}
	\item $g^0(\overline{v})=f(\overline{v})$
	\item $g^{j+1}(\overline{v})=h(\overline{v}, g^j(\overline{v}))$
\end{itemize}

%List terms in $\Delta_0$+Rec are defined as follows. Any $\Delta_0$ term is a $\Delta_0$+Rec list term. 
If $f(\overline{v}), h(\overline{v},y,z)$, and $t(\overline{v})$ are $\Delta_0$+Rec list terms then the expression\newline $Rec[f,h,t](\overline{v})$ is a \textit{recursive term} and its value is given by $g(\overline{v},t)$ with the following definition:
\begin{itemize}
	%\item $g(\overline{v},a)=a$, for any urelement $a$ 
	\item $g(\overline{v},nil)=f(\overline{v})$
	\item $g(\overline{v},cons(\alpha,b))=h(\overline{v},g(\overline{v},\alpha),b)$, for any lists $\alpha,b$ such that $cons(\alpha,b)\sqsubseteq t$.
\end{itemize}

%Note, that Iterative term/\DeltaIt-formula are easily presented by equivalent recursive term/\DeltaRec-formula.

%Let $<i>g(\overline{v},x)$ be the Iterative term. Let $h(cons(\overline{v},x),a,b)=g(\overline{v},a)$, $f(cons(\overline{v},x))=g(\overline{v},x)$ Recursive term $Rec(f,h,i)$ is equivalent to $<i>g(\overline{v},x)$ 	
	
%The list terms introduced above are called \textit{non-standard terms}. 
Let $s$ be a $\Delta_0$+Rec term. $s$ is called \textit{explicit} if in every term $Rec[f,h,t](\overline{v})$, which occurs in $s$, the term $t$ is variable-free. The term $s$ is called \textit{flat} if in every term $Rec[f,h,t](\overline{v})$, which occurs in $s$, $f$ and $h$ are $\Delta_0$-terms. A $\Delta_0$+Rec formula is \textit{flat} if so is every term in it. The notion of flat $\Delta_0$+Iteration term or formula is defined identically. 

Intuitively, the terms, which are not flat, may implement nested looping, which is an additional source of computational complexity. In the paper, we will show however that the complexity of reasoning is increased if $\Delta_0$-formulas are extended only with flat iterative and recursive terms.

Terms $s(\overline{v})$ and $t(\overline{v})$ are called \textit{equivalent} if $HW(\mathcal{M})\models s(\overline{a}) = t(\overline{a})$, for any structure $HW(\mathcal{M})$ and any substitution for $\overline{v}$ with a vector of values $\overline{a}$. 

The \textit{rank} of a $\Delta_0$+Rec term $s$ (notation: $\rank(s)$) is defined as follows. If $s$ is a $\Delta_0$-term then $\rank(s)=0$. If $s=Rec[f,h,t](\overline{v})$ then $\rank(s)$ is the maximum rank of the terms $f,h,t$ increased by $1$. If $s$ is not a recursive term then $\rank(s)$ if the maximum rank of the recursive terms in $s$. The rank of other non-standard terms is defined similarly wrt the maximum rank of the list terms in their parameters. The rank of a formula $\varphi$ (notation: $\rank(\varphi)$) equals to the maximum rank of the  terms in $\varphi$.  

For a list $s$, the \textit{length} of $s$, denoted as $\len(s)$, is the number of elements in $s$, i.e., for $s=\langle t_1\lldots t_n\rangle$ (where every $t_i$, $i=1\lldots n$, is a urelement or a list), we have $\len(s)=n$. %For a $\Delta_0$,  term $t$, the \textit{size} of $t$ is the length of the string, which represents $t$ in the language of $\Delta_0$ formulas. If $t=Rec[f,h,s](\overline{v})$ is a recursive term then the size of $t$ is the total length of the strings representing $f,h,s$ plus the sum of the sizes of the recursive terms in $f,h,s$. The size of an iteration and b-while term are defined similarly.
For a urelement or a list term $t$, the \textit{size} of $t$ (denoted by $|t|$) is the length of the string, which represents $t$. The \textit{size of a formula} (we use the same notation $|\varphi|$) is defined identically. 

%The \textit{size of a formula} $\varphi$ of $\Delta_0$+Rec is the length of the string, which represents $\varphi$, plus the sum of the size of the top-level recursive terms occurring in $\varphi$.  The size of a $\Delta_0$+bWhile and $\Delta_0$+Iteration formulas is defined similarly.
	
\section{Expressiveness of Formulas with Looping Terms}

We begin with an observation that bounded search terms add no expressiveness to $\Delta_0$-formulas in terms of the computational complexity of model checking.

\begin{theorem}[Complexity of Model Checking for $\Delta_0$+bSearch Formulas]
The set of $\Delta_0$+bSearch formulas true in a given structure $HW(\mathcal{M})$ is $\PSPACE$-complete.
\end{theorem}
\begin{proof}
Hardness follows from Theorem 3 in \cite{OspichevPonomaryov}, where it is proved that the set of $\Delta_0$-formulas true in a given structure $HW(\mathcal{M})$ is $\PSPACE$-complete. The upper complexity bound is shown as follows. 

Let $\varphi$ be a $\Delta_0$+bSearch formula and $HW(\mathcal{M})$ a structure. First, we consider the case when every b-search term in $\varphi$ is variable-free and use induction on the rank of $\varphi$ to prove the claim of the theorem. We simultaneously show by induction that there is a $\PSPACE$-computable function, which for a variable-free b-search term $s$ computes the value of $s$ as a list $\underline{s}$ in $HW(\mathcal{M})$ such that the size of $\underline{s}$ is bounded by $|s|$.

For $\rank(\varphi)=\rank(s)=0$ the claims above readily follow from Theorem 3 and Lemma 2 in \cite{OspichevPonomaryov}. For $k=\rank(\varphi)\geqslant 1$, take an arbitrary (variable-free) term $s=bSearch_{\propto}[\theta, t]$, $\propto\in\{\in, \sqsubseteq\}$ of rank $k$ in $\varphi$. Then $t$ is variable-free, $\theta$ has a single free variable, and the ranks of $t$ and $\theta$ are less than $k$. Then by the induction assumption $\underline{t}$ is given by a $\PSPACE$-computable function and the size of $\underline{t}$ is bounded by $|t|$. By applying the induction assumption again, we conclude that there is a $\PSPACE$-computable function, which gives the first element/initial segment $a$ of $\underline{t}$, for which $HW(\mathcal{M})\models\theta(a)$. Clearly, the size of $a$ is bounded by $|s|$.

Now let $\varphi'$ be a $\Delta_0$-formula obtained from $\varphi$ by replacing every b-search term $s$ with $\underline{s}$. By the observation above,  $\varphi'$ can be obtained by a \PSPACE-computable function, it has size bounded by $|\varphi|$, and it holds $HW(\mathcal{M})\models\varphi$ iff $HW(\mathcal{M})\models\varphi'$. Since $\varphi'$ is a $\Delta_0$-formula, we conclude that the claim of the theorem holds for formulas, in which non-standard terms are variable-free.  

For the general case, note that if there is a quantifier $\Game x \propto t $, with $\Game\in\{\exists, \forall\}$, $\propto\in\{\in, \sqsubseteq\}$, in a $\Delta_0$+bSearch formula $\varphi$, then by the definition of $\Delta_0$-formulas, the list term $t$ is variable-free. Let $\varphi'$ be a formula obtained from $\varphi$ by replacing every quantifier of the form $\Game x \propto t $ with $\Game x \propto \underline{t}$, for $\propto\in\{\in, \sqsubseteq\}$. By Lemma 2 in \cite{OspichevPonomaryov} and the above shown, $\varphi'$ can be obtained by a \PSPACE-computable function, it has size bounded by $|\varphi|$, and it holds $HW(\mathcal{M})\models\varphi$ iff $HW(\mathcal{M})\models\varphi'$. Then $HW(\mathcal{M})\models\varphi'$ can be decided by a \PSPACE-computable function. It is given by the standard procedure of bounded quantifier elimination, which stores the selected value for each quantified variable. After all quantifiers are eliminated, a formula $\psi$ from $\varphi'$ is obtained, in which every variable is substituted with the corresponding selected value. The formula $\psi$ is variable-free, thus, by the above shown, $HW(\mathcal{M})\models\psi$ can be decided by a \PSPACE-computable function, from which the claim of the theorem follows. \QED
\end{proof}

In the rest of the paper we focus on the expressiveness of recursive and iterative terms and provide the corresponding complexity results. 

\medskip

Let $L$ be an extension of the language of $\Delta_0$-formulas. %Without loss of generality, we assume further that $L$-formulas are given in the prenex normal form.

%For $k\geqslant 1$ and $n\geqslant 0$ we say that \emph{$\explist{k}{n}$ is expressible in $L$} if there exists a variable-free $L$-term $t$ of size polynomial in $n$ such that for any $L$-structure $HW(\mathcal{M})$, the interpretation of $t$ in $HW(\mathcal{M})$ is a list with $\EXP{k}{n}$-many elements.

%We say that $L$ is a $\kexp$-list extension if it allows to succinctly represent a list of $\EXP{k}{n}$-many elements, for any given $n\geqslant 0$, i.e., it contains (at least a single non-standard) variable-free term $t$ of size polynomial in $n$, which is equal to a standard list term $s$ consisting of $\EXP{k}{n}$-many elements.

For $k\geqslant 0$, we say that a \emph{$k$-list is expressible in $L$} if there exists a variable-free $L$-term $t$ such that for any structure $HW(\mathcal{M})$, the interpretation of $t$ in $HW(\mathcal{M})$ is a list of length $k$.

Let $\times$ be a map, which for non-empty lists $s_1, s_2$ gives a list $\times[s_1, s_2]$, which consists of $conc(\alpha_1,\alpha_2)$, for all $\alpha_i\in s_i$, $i=1,2$. Now let $\circ^k$ be a map defined as follows: for a non-empty list $s$, it holds $s^1=s$ and for $k\geqslant 2$, we have  $s^k=\times[s^{k-1}, s]$.

We say that \emph{$\times$ is expressible in $L$} if there is a $L$-term $t(x_1, x_2)$ such that in any structure $HW(\mathcal{M})$ for any non-empty lists $s_1, s_2$, the term $t(s_1, s_2)$ is interpreted as $\times[s_1, s_2]$. Similarly, for $k\geqslant 1$, $\circ^k$ is said to be expressible in $L$ if there is a $L$-term $t(x)$ such that in any structure $HW(\mathcal{M})$ for any non-empty list $s$, the term $t(s)$ is interpreted as $s^k$. Whenever we want to specify the $L$-term $t$, we say that \textit{$t$ represents} $\times$ ($\cdot^k$, respectively), or $\times$ ($\cdot^k$, respectively) \textit{is expressible by $t$}. We omit a direct reference to the language $L$, whenever it is clear from the context.

\smallskip

\begin{lemma}[Succinctness of Recursive Terms]\label{Lem:SuccinctnessRec}
For $k\geqslant 1$, $n\geqslant 0$, a $\EXP{k}{n}$-list, $\times$, and $\circ^{\EXP{k}{n}}$ is expressible by a recursive term of size linear in $k$ and $n$.
\end{lemma}
\begin{proof}
Let $s_0$ denote the list $\langle nil \rangle$ (i.e., the list, which consists of the single element being the empty list) and for all $n\geqslant 0$, let $s_{n+1}=cons(s_n, nil)$. Given $n\geqslant 0$, we define by induction on $k\geqslant 1$ a variable-free recursive term $\explst_{k}$ as follows. 

For $k=1$, we let $\explst_{1}$ be the term $$Rec[\langle nil \rangle, conc(g(\alpha),g(\alpha)) , s_n].$$ For $k\geqslant 2$, we define $$\explst_{k}=Rec[\langle nil \rangle, conc(g(\alpha),g(\alpha)) , \explst_{k-1}].$$ It easy to verify by induction that the interpretation of $\explst_{k}$ in any structure $\HW{\mathcal{M}}$ is a list, which consists of $\EXP{k}{n}$-many elements (being empty lists). Clearly, the size of $\explst_{k}$ is linear in $k$ and $n$.

\smallskip

Now consider a recursive term, which for any lists $x,y$ gives a list consisting of $conc(x, b)$, for all $b\in y$. It is defined as $$Rec[nil,\  conc(g(\alpha), conc(x, b)), \ y](x, y)$$ We denote this term by $multiply\_element(x,y)$. Now a term which represents $\times$ is defined as $$Rec[nil, \ conc(g(\alpha), \ multiply\_element(b, x_2))), \ x_1](x_1, x_2)$$ Denote it by $multiply(x_1, x_2)$. 

Finally, a recursive term which represents $\circ^{\EXP{k}{n}}$ is given by $$Rec[x, \  multiply(g(\alpha), x ), \ tail(\explst_{k}) ](x)$$ and it is of size linear in $k$ and $n$. \QED
\end{proof}

\begin{lemma}[Succinctness of Iterative terms]\label{Lem:SuccinctnessIteration}
For any $n\geqslant 0$, a $\EXP{1}{n}$-list and $\EXP{2}{n}$-list is expressible by an iterative term of size linear in $n$, in which the number of iterations is given in the unary and binary representation, respectively.
\end{lemma}
\begin{proof}
The proof is identical to Lemma \ref{Lem:SuccinctnessRec}, the variable-free term $$<n>Iter[\langle nil \rangle, conc(g(\alpha),g(\alpha))]$$ is the required one. It represents a $\EXP{1}{n}$-/$\EXP{2}{n}$-list, respectively, if $n$ is given in unary or binary (since the binary representation is exponentially more succinct than the unary one) and its size is linear in $n$. \QED
\end{proof}

\begin{lemma}[Unfolding Lemma]\label{Lem:Unfolding}
For any flat $\Delta_0$+Iteration term $t(\overline{v})$, there is an equivalent $\Delta_0$-term $t_0(\overline{v})$ such that $|t_0|\leqslant \EXP{1}{p(|t|)}$ or $|t_0|\leqslant \EXP{2}{p(|t|)}$, for a polynomial $p$, if the number of iterations is given in unary or binary, respectively. 

For any explicit flat $\Delta_0$+Rec term $t(\overline{v})$ of rank bounded by $k\geqslant 1$, there is an equivalent $\Delta_0$-term $t_0(\overline{v})$, with $|t_0| \leqslant \EXP{k}{p(|t|)}$, for a polynomial $p$.
\end{lemma}
\begin{proof}
Let $t$ be a flat $\Delta_0$+Iteration term. If $t$ is a $\Delta_0$-term, there is nothing to prove, therefore, we assume there is an iterative term $s=<i>Iter[f,h](\overline{v})$ in $t$. We use induction on $i$ to show that the size of a $\Delta_0$-term equivalent to $s$ is bounded by $| f | \cdot | h |^i $. The case $i=0$ is trivial. For $i\geqslant 1$, consider the term $g^i(\overline{v})$ in the definition of $s$. It is given as a combination of a definition for $g^{i-1}(\overline{v})$ with list functions, where $g^{i-1}(\overline{v})$ is equivalent to $<i-1>Iter[f,h](\overline{v})$. The number of occurrences of $g^{i-1}(\overline{v})$ in $g^{i}(\overline{v})$ is at most $|h|$, thus, by applying the induction assumption, the size of a $\Delta_0$-term equivalent to $g^{i}(\overline{v})$ is bounded by $| h |\cdot | f | \cdot | h |^{i-1}$. It follows that the size of the $\Delta_0$-term equivalent to $s$ is bounded by $|s|^{p(|s|)}$ or $|s|^{\EXP{1}{p(|s|)}}$, respectively, for a polynomial $p$, if the number $i$ is given in unary/binary. Hence, it is bounded by $\EXP{1}{p(|s|)}$ or $\EXP{2}{p(|s|)}$, respectively, for an appropriate polynomial $p$. Since the number of iterative terms in $t$ is bounded by $t$ and the choice of $s$ was arbitrary, we conclude that the claim of the lemma holds for $t$. 

Now let $t$ be an explicit flat $\Delta_0+Rec$ term. First, we show that in case $t=Rec[f,h,l_0](\overline{v})$, where $l_0$ is a constant list, there is a $\Delta_0$-term equivalent to $t$ of size bounded by $\EXP{1}{p(|t|)}$, for a polynomial $p$. 

Consider the terms $g(\overline{v},\alpha)$ in the definition of $t$ and denote $t_\alpha=Rec[f,h,\alpha](\overline{v})$, for $\alpha=nil$ or $\alpha\dleq l_0$ We use induction on the length of $l_0$ to show that for any list $\alpha$ such that $\alpha=nil$ or $\alpha\dleq l_0$, there is a $\Delta_0$-term $t^0_\alpha$ equivalent to $t_\alpha$ such that $|t^0_\alpha|$ is bounded by $|f|\cdot |h|^{\len(\alpha)}$. The case $l_0=nil$ is trivial. For $l_0\neq nil$, observe that for all lists $\alpha, b$ the term $g(\overline{v}, cons(\alpha,b))$ in the definition of $t$ is given as a combination of $g(\overline{v}, cons(\alpha))$ with list functions, where $g(\overline{v}, cons(\alpha))$ is equivalent to $t_\alpha$. Thus, by applying the induction assumption, the size of a $\Delta_0$-term equivalent to $g(\overline{v}, cons(\alpha,b))$ is bounded by $| h |\cdot | f | \cdot | h |^{\len(\alpha)}$. It follows that for all $\alpha=nil$ or $\alpha\dleq l_0$, the size of the $\Delta_0$-term equivalent to $t_\alpha$ is bounded by $|t_\alpha|^{p(|t_\alpha|)}$, for a polynomial $p$. Hence, the size of a $\Delta_0$-term equivalent to $t$ is bounded by ${\EXP{1}{p(|t|)}}$, for an appropriate polynomial $p$. 

Now let $t$ be an arbitrary explicit flat $\Delta_0+Rec$ term. We use induction on the rank bound $k\geqslant 1$ for $t$ to show the claim of the lemma. If $t$ is a $\Delta_0$-term then there is nothing to prove. Assume there is a recursive term $s=Rec[f,h,l](\overline{v})$ in $t$, then $f,h$ are $\Delta_0$-terms (since $t$ is flat). If $k=1$ then $\rank(s)=1$ and hence, $l$ is a variable-free $\Delta_0$-term, since $t$ is explicit. By Lemma 2 in \cite{OspichevPonomaryov}, $l$ is equivalent to a constant list $l_0$ of size bounded by a polynomial in the size of $l$ (and hence, in the size of $s$). Then by the above shown, there is a $\Delta_0$-term equivalent to $s$, for which the claim of the lemma holds. 

For $k\geqslant 2$, observe that the term $l$ is given as a combination of recursive terms of rank less than $k$ with list functions. Hence, by the induction assumption, the size of a constant list term $l_0$ equivalent to $l$ is bounded by $\EXP{(k-1)}{p(|s|)}$, for a polynomial $p$. Then by the above shown we conclude that there is a $\Delta_0$-term equivalent to $s$ of size bounded by $\EXP{k}{p(|t|)}$, for a polynomial $p$. Since the choice of $s$ in the term $t$ was arbitrary and the number of non-standard terms in $t$ is bounded by $|t|$, we obtain the required statement. \QED
\end{proof}

\begin{lemma}[Hardness of Model Checking]\label{Lem:ReduceRegularLanguages}
Let $L$ be an extension of the language of $\Delta_0$-formulas such that $\times$ is expressible in $L$ and for all $n\geqslant 0$ and some $k\geqslant 1$,  $\circ^{\EXP{k}{n}}$ is expressible by a $L$-term of size polynomial in $k$ and $n$. Then the set of $L$-formulas true in a given structure $HW(\mathcal{M})$ is $\NmanyEXPTIME{k}$-hard.
\end{lemma}
\begin{proof}
The lemma is proved by a reduction of the inequality problem for regular-like expressions without the Kleene star, but with the exponentiation operation. Regular-like expressions of this kind are defined over a finite alphabet $\Sigma$ by using the operation of union $\cup$, concatenation $\cdot$, and exponentiation $\cdot^{\EXP{k}{n}}$, where $k\geqslant 1$ and $n\geqslant 0$. For a regular-like expression $E$, the language $L(E)$ is given inductively as a subset of all strings over $\Sigma$ by the following definition: 
\begin{itemize}
\item $L(a)=\{a\}$, for $a\in\Sigma$
\item $L(E_1\cup E_2)=L(E_1)\cup L(E_2)$
\item $L(E_1\cdot E_2) = \{s_1\cdot s_2 \mid s_i\in L(E_i), \ i=1,2\}$
\item $L(E^{\EXP{k}{n}}) = \underbrace{L(E)\cdot\ldots \cdot L(E)}_{\EXP{k}{n}~\text{\textit{times}}}$
\end{itemize}

The size of a regular-expression is the length of the string, which represents it. The inequality problem for regular-like expressions is the set of pairs $\langle E_1, E_2 \rangle$ such that $L(E_1)\neq L(E_2)$. It is shown in \cite{Stockmeyer} that the inequality problem is $\NEXPTIME$-hard for regular-like expressions, in which exponentiation is restricted to $\EXP{1}{n}$, for $n\geqslant 0$. The proof employs a direct reduction from the halting problem for non-deterministic Turing machines making at most $\EXP{1}{n}$-many steps on an input of size $n\geqslant 0$. For the reduction, the subsequent configurations of the TM are encoded by regular-like expressions over an alphabet $\Sigma$ of the TM , which represent the tape content and the state of the TM. In particular, each configuration is represented as a word of length $2\cdot\EXP{1}{n}+1$. Expressions of the form $(\sigma)^{f(n)}$ are used to refer to (parts of) configurations of TM, where $\sigma\subseteq\Sigma\cup\{\lambda\}$ and $f(n)=c_0\cdot\EXP{1}{n}+nc_1+c_2$, with $c_0\geqslant 0$ and $c_1, c_2$ being integers such that $|c_i|\leqslant n$, for $i=0,1,2$.

If instead, each configuration is represented by a word of length  $2\cdot\EXP{1}{n}+n+2$ (which can be made without loss of generality), then every regular-like expression of the form $(\sigma)^{f(n)}$, with $f(n)$ as above, can be replaced by $(\sigma)^{g(n)}$, where $g(n)=c_0\cdot\EXP{1}{n}+nd_1+d_2$ is a function, with $0 \leqslant c_0,d_1, d_2\leqslant n$. Then the proof works for exponentiation $\EXP{k}{n}$ with any $k\geqslant 1$, $n\geqslant 0$ and gives $\NmanyEXPTIME{k}$-hardness of the inequality problem for regular-like expressions with the operation $\cdot^{\EXP{k}{n}}$. We reduce this problem to checking the truth of $L$-formulas.

Let $t_\times(x)$ be a $L$-term, which represents $\times$, and $t_{\EXP{k}{n}}(x_1,x_2)$ a $L$-term of size polynomial in $k$ and $n$, which represents $\circ^{\EXP{k}{n}}$. For a regular-like expression $E$ we inductively define the $L$-term $list(E)$, which encodes the language $L(E)$ as:
\begin{itemize}
\item $list(\{a\}) = \langle\langle a \rangle\rangle$, for $a\in\Sigma$;
\item $list(E_1\cup E_2) = conc(list(E_1), \ list(E_2))$;
\item $list(E_1\cdot E_2) = t_\times(list(E_1), \ list(E_2))$;
\item $list(E^{\EXP{k}{n}}) = t_{\EXP{k}{n}}(list(E))$.
\end{itemize}

Clearly, the size of $list(E)$ is linear in the size of the expression $E$. 

Now it suffices to note that for any structure $HW(\mathcal{M})$ and any regular-like expressions $E_1, E_2$, it holds $L(E_1)\neq L(E_2)$ iff $$HW(\mathcal{M})\models \exists x\in list(E_1) (\ x\not\in list(E_2) \ ) \ \vee \  \exists x\in list(E_2) (\ x\not\in list(E_1) \ )$$ Indeed, if there is such $x$ then it has the form $\langle a_1\lldots a_k\rangle$, where $k\geqslant 1$, $a_i\in\Sigma$, and then the word $a_1\ldots a_k$ witnesses the difference between $L(E_1)$ and $L(E_2)$. The opposite direction is straightforward. \QED
\end{proof}

\begin{theorem}[Complexity of Model Checking for $\Delta_0$+Rec]\label{Thm:ModelCheckingRec}
The set of $\Delta_0$+Rec formulas true in a given structure $HW(\mathcal{M})$ contains the class $\manyEXPTIME{k}$, for every $k\geqslant 1$, and hence, it is non-elementary. It follows that there is a $\Delta_0$+Rec formula $\varphi$, which is not equivalent to a $\Delta_0$-formula $\psi$ of size polynomial in $|\varphi|$.
\end{theorem}
\begin{proof}
By Lemma \ref{Lem:SuccinctnessRec}, $\times$ is expressible in $\Delta_0$+Rec and for all $k\geqslant 1$ and $n\geqslant 0$, $\circ^{\EXP{k}{n}}$ is expressible by a $\Delta_0$+Rec term of size linear in $k$ and $n$. Then by Lemma \ref{Lem:ReduceRegularLanguages}, for any $k\geqslant 1$, there exists a $\NmanyEXPTIME{k}$-hard set of $\Delta_0$+Rec formulas true in a given $HW(\mathcal{M})$. For all $k\geqslant 1$, it holds $\manyEXPTIME{k}\subseteq\NmanyEXPTIME{k}$ and hence, the set of $\Delta_0$+Rec formulas true in $HW(\mathcal{M})$ is non-elementary. 

It was proved in \cite{OspichevPonomaryov} that the set of $\Delta_0$-formulas true in a given structure $HW(\mathcal{M})$ is $\PSPACE$-complete. Assume that for any $\Delta_0$+Rec formula $\varphi$, there is an equivalent $\Delta_0$-formula $\psi$ of size polynomial in the size of $\varphi$. Then $\cup_{k\geqslant 1}\manyEXPTIME{k}\subseteq\PSPACE$, which is a contradiction, since already $\DEXPTIME$ is not contained in $\PSPACE$. \QED
\end{proof}

\begin{theorem}[Complexity of Model Checking for Flat $\Delta_0$+Iteration]\label{Thm:ModelCheckingIteration}
The set of flat $\Delta_0$+Iteration formulas true in a given structure $HW(\mathcal{M})$ is $\PSPACE$-hard and it is in $\EXPTIME$ or $\DEXPTIME$ if the number of iterations is given in unary or binary, respectively.
\end{theorem}
\begin{proof}
The lower complexity bound follows from Theorem 3 in \cite{OspichevPonomaryov}, where it is shown that the set $\Delta_0$-formulas true in a given structure $HW(\mathcal{M})$ is $\PSPACE$-complete. The upper bound is shown as follows. We consider the case, when the number of iterations in formulas is given in unary, the proof for the binary case is identical.

Let $HW(\mathcal{M})$ be a structure and $\varphi$ a flat $\Delta_0$+Iteration formula of the form $$\Game_1 x_1 \propto_1 t_1 \ \ldots \ \Game_n x_n \propto_n t_n \ \ \psi(x_1\lldots x_n)$$ where $n\geqslant 0$, $\Game_i\in\{\exists, \forall\}$, and $\propto_i \in\{\in, \sqsubseteq\}$, for all $i=1\lldots n$. %We assume without loss of generality that $\varphi$ contains at least one iterative term. 
Let $T$ denote the total size of the terms $t_1\lldots t_n$.  We will show by induction on the complexity of $\varphi$ that satisfiability of $\varphi$ can be decided by at most $\EXP{1}{p(|T|)}$-many tests for satisfiability of formulas $\psi'$ obtained from $\psi$ by substitutions of $x_1\lldots x_n$ with vectors of constant lists, each of size bounded by $\EXP{1}{p(T)}$, where $p$ is a polynomial. We refer to this claim further as $(*)$.

Then it follows from the proof of Lemma \ref{Lem:Unfolding} that satisfiability of $\varphi$ can be decided by at most $\EXP{1}{p(T)}$-many tests for satisfiability of $\Delta_0$-formulas, each of which is either $\psi'$ as above (if $\psi$ does not contain non-standard terms) or obtained from $\psi'$ as a value of a $\EXPTIME$-computable function and has size bounded by $\EXP{1}{r(|\varphi|)}$, for some polynomials $p, r$. As $T\leqslant |\varphi|$ and the set of variable-free $\Delta_0$-formulas true in a given structure $HW(\mathcal{M})$ is $\PTIME$-computable, we obtain the statement of the theorem.

We now show that claim $(*)$ holds. If $\varphi$ is quantifier-free, there is nothing to prove. Now assume $\varphi$ has the form $\Game x \propto t \ \theta(x)$, where $\Game\in\{\exists, \forall\}$ and $\propto \in\{\in, \sqsubseteq\}$. It is equivalent to the formula $\varphi' = \Game x \propto t_0 \ \theta(x)$, where $t_0$ is a constant term equivalent to $t$. It follows from the proof of Lemma \ref{Lem:Unfolding} that $t_0$ can be obtained as a value of a $\EXPTIME$-computable function. In particular, the number and the size of lists $a\propto t_0$ is bounded by $\EXP{1}{p(T)}$, where $p$ is a polynomial. Then satisfiability of $\varphi$ can be decided by $\EXP{1}{p(T)}$-many tests of satisfiability of $\theta(a)$, one for each $a\propto t_0$, and thus, by applying the induction assumption to the formulas $\theta(a)$, we obtain the required claim. \QED
\end{proof}

We now turn to the complexity results on satisfiability. Note that Theorem \ref{Thm:ModelCheckingRec} provides a lower bound on the complexity of testing satisfiability of $\Delta_0$+Rec, while Theorem \ref{Thm:ModelCheckingIteration} does not provide any lower bound (other than \PSPACE, which is known already for $\Delta_0$-formulas). However, it is possible to obtain tight complexity results by using the reduction, which we describe next.

\begin{lemma}[Hardness of Satisfiability]\label{Lem:ReduceDomino}
Let $L$ be an extension of the language of $\Delta_0$-formulas such that for all $n\geqslant 0$ and some $k\geqslant 1$, a $\EXP{k}{n}$-list is expressible in $L$ by a $L$-term of size polynomial in $k$ and $n$. Then the set of satisfiable $L$-formulas is $\NmanyEXPTIME{k}$-hard.
\end{lemma}
\begin{proof}
The lemma is proved by a reduction from the (bounded) domino tiling problem \cite{Gurevich}. A \emph{domino system} is a triple $\D=(T,V,H,init)$, where $T=\{1,\ldots ,p\}$, for $p\geqslant 1$, is a finite set of tiles, $H,V\subseteq T\times T$ are horizontal and vertical tile matching relations, and $init=\tuple{t_1,\ldots , t_s}$ is an initial tiling condition, where $t_i\in T$, for $1\leqslant i \leqslant s$, and $s\geqslant 0$. A \emph{tiling} of size $m\times m$ for a domino system $\D$ is a mapping $t:\{1,\ldots ,m\}\times\{1,\ldots ,m\}\rightarrow T$ such that $\tuple{t(y-1,x), \ t(y,x)}\in V$, for $1 < y \leqslant m$, $1 \leqslant x \leqslant m$, $\tuple{t(y,x-1), \ t(y,x)}\in H$, for $1 \leqslant y \leqslant m$, $1 < x \leqslant m$, and $t(1,x)=t_x$, for $1\leqslant x \leqslant s$. The size of  a domino system is measured as $s$ plus the sum of the cardinalities of $V,H$, and $T$. It is known that the set of domino systems, which admit a tiling of size $\EXP{k}{n}\times\EXP{k}{n}$, where $k\geqslant 1$, $n\geqslant 0$, is $\NmanyEXPTIME{k}$-complete.

Let $\D$ be a domino system, $k\geqslant 1$, $n\geqslant 0$, and let $\explst$ be a $L$-term of size polynomial in $k$ and $n$, which represents a $\EXP{k}{n}$-list. We define a set of $L$-formulas $\bt$ with quantification over elements and initial segments of $\explst$, which encode the tiling problem for $\D$ and a grid of dimension $\EXP{k}{n}\times\EXP{k}{n}$  to be ``tiled''. We assume without loss of generality that $s\leqslant \EXP{k}{n}$.

The theory $\bt$ is defined over a signature $\Sigma$, which contains a binary predicate $T_i$, for every tile $i\in\{1,\lldots p\}$. In particular, it includes predicates $t_1,\ldots , t_s$ corresponding to the tiles in the initial condition. In our encoding of the tiling problem, we represent an element of a grid of an exponential size by a pair of lists being initial segments of $\explst$ (there are $\EXP{k}{n}$-many of them), which corresponds to the ``coordinate'' of the grid element. 

\smallskip

First of all, the theory $\bt$ contains axioms 

\begin{equation}\label{Eq:AtLeast1Tile}
\forall x,y\sqsubseteq\explst \ \bigvee_{i\in T} T_i(\ y, x\ )
\end{equation}

and

\begin{equation}\label{Eq:AtMost1Tile}
\forall x,y\sqsubseteq\explst \  \neg (\ T_i(\ y, x \ )\wedge T_j( \ y,x \ ) \ )
\end{equation}

for all distinct $i,j\in T$. 

\smallskip

These axioms state that every element of the grid is ``occupied'' by exactly one tile.

\smallskip

The next axiom encodes the initial tiling condition $\langle t_1\lldots t_s\rangle$ and we assume that it is present in $\bt$ if $s\geqslant 1$:

\begin{equation}\label{Eq:InitialCondition}
t_1(\explst,\explst)\wedge t_2(\explst,tail(\explst)) \wedge \ldots\wedge t_s(\explst,\underbrace{tail(tail(\ldots tail}_{s-1~\text{times}}{}(\explst))\ldots)
\end{equation}

%The next axioms of $\bt$ encode the initial tiling condition $\langle t_1\lldots t_s\rangle$ and we assume that they are present in $\bt$ if $s\geqslant 1$: 

%\begin{equation}\label{Eq:InitialCondition}
%\begin{aligned}
%\exists \ x,y &\in\explst \ (\ element_{1,1}(y,x) \ \wedge \ t_1(y, x) \ ) \\
%\exists \ x,y &\sqsubseteq\explst \ (\ element_{1,2}(y,x) \ \wedge \ t_2(y, x) \ ) \\
%& \ldots \\
%\exists \ x,y &\sqsubseteq\explst \ (\ element_{1,s}(y,x) \ \wedge \ t_s(y, x) \ )
%\end{aligned}
%\end{equation}

%where $element_{1,1}(y,x)$ stands for $$cons(nil, x)\sqsubseteq\explst \wedge cons(nil, y)\sqsubseteq\explst$$ and for all $i\in\{2\lldots s\}$, $element_{1,i}(y,x)$ is a shortcut for $$cons(\ nil, \ head(\underbrace{tail(tail(\ldots tail(x))\ldots)}_{i-1~\text{times}}{}) \ ) \sqsubseteq \explst$$.

\smallskip

The following axioms represent the vertical matching condition on tiling:

\begin{equation}\label{Eq:VMatching}
\forall \ x,y_1,y_2\dleq\explst \ \  \neg(\ y_1=tail(y_2) \wedge T_i(\ y_1, x\ )) \wedge T_j(\ y_2, x \ )) \ ) 
\end{equation}
for all $i,j\in T$ such that $\tuple{j,i}\not\in V$. 

\smallskip

Finally, the next axioms represent the horizontal matching condition:

\begin{equation}\label{Eq:HMatching}
\forall \ x_1,x_2,y\dleq\explst \ \ \neg (\ x_1=tail(x_2) \wedge T_i(\ y,x_1 \ ) \wedge T_j( \ y,x_2 \ ) \ ) 
\end{equation}
for all $i,j\in T$ such that $\tuple{j,i}\not\in H$. 

\medskip

The definition of the theory $\bt$ is complete. 

\medskip

It is easy to see that the size of $\bt$ is polynomially bounded by the size of the domino system $D$. We claim that $D$ admits a tiling of size $\EXP{k}{n}\times\EXP{k}{n}$ iff there is a structure $\mathcal{M}$ such that $HW(\mathcal{M})\models\bt$. 

For a list $s$ and $1\leqslant z \leqslant \len(s)$, let $seg^z(s)$ denote the initial segment of $s$, which consists of $(\len(s)-z+1)$-many elements. 

$(\Leftarrow)$: Given a model $HW(\mathcal{M})$ of $\bt$, define a mapping $t:\{1,\ldots ,\EXP{k}{n}\}\times\{1,\ldots ,\EXP{k}{n}\}\rightarrow T$ by setting $t(y,x)=k$ iff $HW(\mathcal{M})\models T_k(seg^y(\explst), seg^x(\explst))$. By axioms \ref{Eq:AtLeast1Tile},\ref{Eq:AtMost1Tile}, the mapping $t$ is well defined. By axiom \ref{Eq:InitialCondition}, it respects the initial tiling condition and by axioms \ref{Eq:VMatching},\ref{Eq:HMatching} it satisfies the vertical and horizontal matching conditions. Thus, the mapping $t$ is a tiling.

$(\Rightarrow)$: Given a tiling $t$, consider a structure $HW(\mathcal{M})$ of signature $\Sigma$, in which $\explst$ is interpreted as some list $s$ (of length $\EXP{k}{n}$) and the binary predicates are interpreted as follows: for any tile $k\in T$ and lists $l_1, l_2$, it holds $HW(\mathcal{M})\models T_k(l_1, l_2)$ iff $l_1=seg^y(s)$ and $l_2= seg^x(s))\rangle$, for some $1\leqslant x,y \leqslant \EXP{k}{n}$ and $t(y,x)=k$. Since $t$ is a map, the structure $HW(\mathcal{M})$ defined this way is a model of axioms \ref{Eq:AtLeast1Tile},\ref{Eq:AtMost1Tile}. As $t$ respects the initial condition and matching conditions, $HW(\mathcal{M})$ is a model of axioms \ref{Eq:InitialCondition}-\ref{Eq:HMatching} and hence, it is a model of $\bt$. \QED
\end{proof}

Now we are in the position to formulate the complexity results on satisfiability of flat $\Delta_0$-formulas extended with iterative or recursive terms. 

\begin{theorem}[Complexity of Satisfiability for Flat $\Delta_0$+Iteration]\label{Thm:SatisfiabilityIteration}
The set of satisfiable flat $\Delta_0$+Iteration formulas is $\NEXPTIME$-complete if the number of iterations is given in unary and it is $\NDEXPTIME$-complete if the number of iterations is given in binary.
\end{theorem}
\begin{proof}
Hardness follows from Lemma \ref{Lem:ReduceDomino} and the construction from the proof of Lemma \ref{Lem:SuccinctnessIteration}, where it is shown that for any $n\geqslant 0$ and $k = 1$ or $k=2$, respectively, a $\EXP{k}{n}$-list is expressible by a flat iterative term, in which the number of iterations is given in unary/binary. The upper complexity bound is shown by a repetition of the proof of Theorem \ref{Thm:ModelCheckingIteration} and by using the fact that the set of satisfiable variable-free $\Delta_0$-formulas is in $\NP$. \QED
\end{proof}

%It was shown in \cite{Doklady} that 

\begin{theorem}[Complexity of Satisfiability for Flat $\Delta_0$+Rec]\label{Thm:SatisfiabilityRestrictedRec}
The set of satisfiable flat $\Delta_0$+Rec formulas, which contain at most $k\geqslant 1$ recursive terms, is $\NmanyEXPTIME{k}$-complete.
\end{theorem}
\begin{proof}
Hardness follows from Lemma \ref{Lem:ReduceDomino} and the construction from the proof of Lemma \ref{Lem:SuccinctnessRec}, which shows that for any $k\geqslant 1$ and $n\geqslant 0$, a $\EXP{k}{n}$-list is expressible by a flat term, which contains $k$ recursive terms. The proof for the upper bound is similar to the proof of Theorem \ref{Thm:ModelCheckingIteration}.

Let $\varphi$ be a flat $\Delta_0$+Rec formula of the form $$\Game_1 x_1 \propto_1 t_1 \ \ldots \ \Game_n x_n \propto_n t_n \ \ \psi(x_1\lldots x_n)$$ where $n\geqslant 0$, $\Game_i\in\{\exists, \forall\}$, and $\propto_i \in\{\in, \sqsubseteq\}$, for all $i=1\lldots n$. %We assume without loss of generality that $\varphi$ contains at least one recursive term. 
Let $T$ denote the total size of the terms $t_1\lldots t_n$.  We will show by induction on the complexity of $\varphi$ that satisfiability of $\varphi$ can be decided by at most $\EXP{m}{p(|T|)}$-many tests for satisfiability of formulas $\psi'$ obtained from $\psi$ by substitutions of $x_1\lldots x_n$ with a vector of constant lists, each of size bounded by $\EXP{m}{p(T)}$, where $p$ is a polynomial and $m=max(\rank(t_1)\lldots\rank(t_n))$. We refer to this claim further as $(*)$. It yields the statement of the theorem due to the following observation. 

Since $m$ is bounded by the number of recursive terms in the quantifier prefix of $\varphi$, the number of recursive terms in $\psi'$ (and hence, their rank) is less or equal than $k-m$. Then it follows from the proof of Lemma \ref{Lem:Unfolding} that satisfiability of $\varphi$ can be decided by at most $\EXP{m}{p(T)}$-many tests for satisfiability of $\Delta_0$-formulas, each of which is either $\psi'$ as above (if $\psi$ does not contain non-standard terms) or obtained from $\psi'$ as a value of a $\manyEXPTIME{k}$-computable function and has size bounded by $\EXP{k}{r(|\varphi|)}$, for some polynomials $p, r$. Since $m\leqslant k$, $T\leqslant |\varphi|$,  and the set of satisfiable variable-free $\Delta_0$-formulas is in $\NP$, we obtain the statement of the theorem.

Let us now show that $(*)$ holds. If $\varphi$ is quantifier-free, there is nothing to prove. Now assume that $\varphi$ has the form $\Game x \propto t \ \theta(x)$, where $\Game\in\{\exists, \forall\}$ and $\propto \in\{\in, \sqsubseteq\}$. It is equivalent to the formula $\varphi' = \Game x \propto t_0 \ \theta(x)$, where $t_0$ is a constant list term equivalent to $t$. It follows from the proof of Lemma \ref{Lem:Unfolding} that $t_0$ can be obtained as a value of a $\manyEXPTIME{m}$-computable function, where $m$ is the maximal rank of terms $t_i$, $i=1\lldots n$, in the quantifier prefix of $\varphi$. In particular, the number and the size of every list $a\propto t_0$ is bounded by $\EXP{m}{p(T)}$, for a polynomial $p$. Then satisfiability of $\varphi$ can be decided by $\EXP{m}{p(T)}$-many tests of satisfiability of $\theta(a)$, one for each $a\propto t_0$, and thus, by applying the induction assumption to the formulas $\theta(a)$, we obtain the claim. \QED
\end{proof}

\section{Conclusions}
We have shown that looping terms can succinctly represent exponentially long lists and can express Cartesian concatenation of lists, which may be the source of the increased computational complexity. For the latter operation, nested iteration over lists is required. If the number of iterations is bounded by some number $k$, then it is possible to implement Cartesian concatenation via iterative terms only for lists of $k$-bounded length. Thus, there remains a certain gap in understanding the expressiveness of iterative terms. On one hand, the can succinctly represent exponentially long lists (even when terms are flat and the number of iterations is given in unary), on the other hand, they allow for expressing Cartesian concatenation of lists only of polynomially bounded length. We leave it open whether the lower bound on the complexity of model checking for flat $\Delta_0$-formulas with iterative terms matches the upper bound shown in this paper. We have proved tight complexity bounds for satisfiability, which hint to the natural connection with model checking in terms of complexity. If model checking is complete in some complexity class (e.g., $\EXPTIME$), then typically satisfiability is complete in the non-deterministic variant of this class (i.e., $\NEXPTIME$), and vice versa. In the paper we considered extensions of the language of $\Delta_0$-formulas with bounded search, iterative, and recursive terms as separate languages. In further research, we plan to study the interplay between non-standard terms, when they are used in formulas simultaneously, and to identify syntactic restrictions on the form of terms and formulas which guarantee tractable reasoning.

\end{document}